\RequirePackage{amsthm}
\documentclass[sn-mathphys-num,pdflatex]{sn-jnl}

\usepackage{graphicx}%
\usepackage{multirow}%
\usepackage{amsmath,amssymb,amsfonts}%
\usepackage{amsthm}%
\usepackage{mathrsfs}%
\usepackage[title]{appendix}%
\usepackage{xcolor}%
\usepackage{textcomp}%
\usepackage{manyfoot}%
\usepackage{booktabs}%
\usepackage{listings}%

\usepackage{tabularray}
\usepackage{overpic}

\theoremstyle{thmstyleone}%
\newtheorem{theorem}{Theorem}
\newtheorem{corollary}[theorem]{Corollary}%
\newtheorem{lemma}[theorem]{Lemma}%
\newtheorem{construction}{Construction}

\theoremstyle{thmstyletwo}%
\newtheorem{example}{Example}%

\theoremstyle{thmstylethree}%
\newtheorem{definition}{Definition}%

\usepackage{mathtools}
\DeclarePairedDelimiter\abs{\lvert}{\rvert}

\DeclarePairedDelimiter\set{\{}{\}}

\renewcommand{\leq}{\leqslant}

\renewcommand{\geq}{\geqslant}

\newcommand{\cT}{\mathcal{T}}

\newcommand{\N}{\mathbb{N}}

\newcommand{\Z}{\mathbb{Z}}

\newcommand{\oa}{\overline{a}}
\newcommand{\ob}{\overline{b}}

\newcommand{\ov}{\overline{v}}
\newcommand{\ox}{\overline{x}}

\newcommand{\eqdef}{\triangleq}
\newcommand{\der}{\Longrightarrow}

\newcommand{\rc}{\mathrm{rc}}
\newcommand{\pal}{\mathrm{pal}}
\newcommand{\sA}{\mathsf{A}}
\newcommand{\sC}{\mathsf{C}}
\newcommand{\sG}{\mathsf{G}}
\newcommand{\sT}{\mathsf{T}}

\begin{document}

\title[On the Coding Capacity of Reverse-Complement and Palindromic Duplication-Correcting Codes]{On the Coding Capacity of Reverse-Complement and Palindromic Duplication-Correcting Codes}


\author*[1,2]{\fnm{Lev} \sur{Yohananov}}\email{levyuhananov@gmail.com}

\author[1,3]{\fnm{Moshe} \sur{Schwartz}}\email{schwartz.moshe@mcmaster.ca}

\affil[1]{\orgdiv{School of Electrical and Computer Engineering}, \orgname{Ben-Gurion University of the Negev}, \orgaddress{\city{Beer Sheva}, \postcode{8410501}, \country{Israel}}}

\affil[2]{\orgdiv{Department of Electrical and Computer Engineering, Institute for Systems Research}, \orgname{University of Maryland}, \orgaddress{\city{College Park, MD}, \postcode{20742}, \country{USA}}}

\affil[3]{\orgdiv{Department of Electrical and Computer Engineering}, \orgname{McMaster University}, \orgaddress{\city{Hamilton}, \postcode{L8S 4K1}, \state{ON}, \country{Canada}}}


\abstract{
We derive the coding capacity for duplication-correcting codes capable of correcting any number of duplications. We do so both for reverse-complement duplications, as well as palindromic (reverse) duplications. We show that except for duplication-length $1$, the coding capacity is $0$. When the duplication length is $1$, the coding capacity depends on the alphabet size, and we construct optimal codes.
}




\maketitle

\section{Introduction}

Storing information in DNA molecules promises a medium that is six orders of magnitude denser than the densest electronic storage medium of today. The feasibility of this technology was first demonstrated in~\cite{ChuGaoKos12,GolBerCheDesLepSipBri13} (for in-vitro DNA storage), and then for storage in the DNA of living organisms, in~\cite{ShiNivMacChu17} (for in-vivo DNA storage). Apart from plain storage, the latter also allows watermarking of genetically-modified organisms (GMOs), labeling organisms for biological studies, and running synthetic-biology algorithms.

As with any storage or communication system, the transmitted information is corrupted by the channel and received with errors. In the case of DNA storage, some of the observed error mechanisms have already been studied in the literature in other contexts, mainly electronic communications. Among these we can see substitution errors, insertions and deletions. However, in-vivo DNA storage exhibits a new kind of errors -- duplication errors -- which may be caused during faulty DNA replication.

A DNA strand is a long string over the alphabet $\set{\sA,\sC,\sG,\sT}$, standing for the four chemical bases: adenine, cytosine, guanine, and thymine. The bases are partitioned into complement pairs, where $\sA$ and $\sT$ are complement to each other, as are $\sC$ and $\sG$. In a nutshell, a duplication error takes a substring of the DNA string, and inserts a copy of it (perhaps altered) back into the string. Possible alterations vary, depending on the biological process creating them. Common examples we shall be referring to throughout the paper, include \emph{tandem duplication}, where the copy is unaltered and placed next to its original location, e.g.,
\[
\sA\sA\sC\sT\sG\sG\sA\sT\sC\sC\sC\sT \der \sA\sA\sC\sT\sG\sG\sA\sT\underline{\sG\sG\sA\sT}\sC\sC\sC\sT,
\]
where a duplication of length $4$ occurred, inserting a duplicate copy of $\sG\sG\sA\sT$ (underlined on the right) immediately following its original location. Another type of duplication is \emph{palindromic duplication} (sometimes called \emph{reverse duplication}), which alters the inserted copy by reversing it, e.g.,
\[
\sA\sA\sC\sT\sG\sG\sA\sT\sC\sC\sC\sT \der \sA\sA\sC\sT\sG\sG\sA\sT\underline{\sT\sA\sG\sG}\sC\sC\sC\sT.
\]
Finally, we mention \emph{reverse-complement duplication}, in which the altered copy is both reversed and complemented, e.g.,
\[
\sA\sA\sC\sT\sG\sG\sA\sT\sC\sC\sC\sT \der \sA\sA\sC\sT\sG\sG\sA\sT\underline{\sA\sT\sC\sC}\sC\sC\sC\sT.
\]

The string-duplication channel was introduced in~\cite{FarSchBru16}, and the first duplication-correcting codes were developed in~\cite{JaiFarSchBru17a}. These generalize the setting to allow for arbitrary alphabets, parametric duplication length, and various duplication types. Codes correcting any number of tandem duplications were studied in~\cite{JaiFarSchBru17a,ZerEsmGul19,ZerEsmGul20,YuSch24} and any number of palindromic duplications in~\cite{ZerEsmGul20}. Codes that correct only a fixed number of tandem duplications (sometimes just one) were studied in~\cite{Kov19,LenWacYaa19,GosPolVor23}, palindromic duplications in~\cite{NguCaiSonImm22}, and reverse-complement duplications in~\cite{BenSch22}. Additionally, codes that correct mixtures of duplications, substitutions, as well as insertions and deletions, have been the focus of~\cite{TanYehSchFar20,TanFar21a,TanWanLouGabFar23}. We also mention related work on duplication errors, though not providing error-correcting codes, but rather studying the possible outcomes from multiple errors~\cite{FarSchBru16,JaiFarBru17,EliFarSchBru19,FarSchBru19,BenSch22,Eli24}.

In this paper we focus on the asymptotic rate of optimal codes (the \emph{coding capacity}) that are capable of correcting any number of duplication errors. The coding capacity for tandem-duplication errors is already known~\cite{JaiFarSchBru17a}. Thus, we study reverse-complement duplications, and palindromic duplications. We provide a complete solution, for all duplication lengths $k$, and all possible alphabet sizes $q$. We show that, except for the case of duplication length $k=1$, the coding capacity is $0$. When $k=1$, the coding capacity depends on the alphabet size $q$, and in this case we construct optimal codes. A summary of the results is provided in Table~\ref{tab:summary}.

The paper is organized as follows. In Section~\ref{sec:pre} we give all the necessary notation and definitions. In Section~\ref{sec:rcbin1} we study reverse-complement duplication with duplication length $k=1$, whereas in Section~\ref{sec:rclong} we focus on longer duplication lengths. Section~\ref{sec:pal} provides the results for palindromic duplication. We conclude in Section~\ref{sec:conc} with a short summary of the results and some open questions.

\begin{table}[t]
\caption{A summary of the coding-capacity results, where $k$ denotes the duplication length, $q$ denotes the alphabet size, and $R_q(*)_k$ denotes the coding capacity. Note that for reverse-complement duplication, $q$ must be even.}
\label{tab:summary}
\begin{tabular}{ccccl}
\toprule%
Type & $k$ & $q$ & $R_q(*)_k$ & Location  \\
\midrule
\multirow{3}{*}{Reverse Complement} &  \multirow{2}{*}{$=1$}       & $=2$      & $0$ & Corollary~\ref{cor:binary1} \\
\cmidrule{3-5}
    &           & $\geq 4$  & $\log_q(q-2)$ & Corollary~\ref{cor:nonbinary1} \\
    \cmidrule{2-5}
    & $\geq 2$  & $\geq 2$       & $0$ & Corollary~\ref{cor:long}\\
\midrule
\multirow{2}{*}{Palindromic} & $=1$          & $\geq 2$  & $\log_q(q-1)$ & \cite{JaiFarSchBru17a}, equivalent to tandem duplication \\
\cmidrule{2-5}
    & $\geq 2$  & $\geq 2$       & $0$ & Corollary~\ref{cor:pallong}\\\botrule
\end{tabular}
\end{table}

\section{Preliminaries}
\label{sec:pre}

Throughout this paper, we shall be considering strings over a finite alphabet. Without loss of generality, we shall assume that alphabet is $\Z_q$, the ring of integers modulo $q$, where $q\in\N$, $q\geq 2$. A \emph{string} is a sequence of \emph{letters} from the alphabet $\Z_q$. The set of strings of length $n\in\N$ is denoted by $\Z_q^n$. Assume $x\in\Z_q^n$ is such a string, then the letters of $x$ are denoted by adding a subscript $i$, where indices start from $0$. Thus, $x=x_0 x_1 \dots x_{n-1}$, with $x_i\in\Z_q$ for all $i$. We denote the \emph{length} of $x$ by $\abs{x}=n$. If $y\in\Z_q^m$ is also a string, we denote by $xy$ the \emph{concatenation} of $x$ and $y$. If $\ell\geq 0$ is an integer, we denote by $x^\ell$ the concatenation of $\ell$ copies of $x$. Since subscripts are used to denote individual letters in a string, and superscripts denote concatenation of copies of a string, when we require several strings with the same name, we shall use superscripts in parentheses to distinguish between them, e.g., $x^{(0)}, x^{(1)}, x^{(2)},\dots$.

The set of all finite-length strings over $\Z_q$ is denoted by $\Z_q^*$. We use $\varepsilon$ to write the unique \emph{empty string} of length $0$. Obviously $\varepsilon\in\Z_q^*$. If we are interested only in strings with positive length we write $\Z_q^+\eqdef\Z_q^*\setminus\set{\varepsilon}$.

Assume $x=uvw$, with $u,v,w\in\Z_q^*$, $\abs{v}=k$ for some $k\in\N$. Then we say that $v$ is a \emph{$k$-factor} of $x$. If $u=\varepsilon$, then $v$ is said to be a \emph{$k$-prefix} of $x$, and if $w=\varepsilon$ then $v$ is said to be a \emph{$k$-suffix} of $x$. If the length of $x$ is not of importance, we just say it is a factor, or prefix, or suffix of $x$.

The notion of \emph{complement} will be important. We assume the letters of the alphabet $\Z_q$ may be partitioned into pairs of complements. If $a\in\Z_q$ is a letter, we denote its complement by $\oa$, where necessarily $\oa\neq a$. Additionally, $\overline{\oa}=a$. It follows that if a complement is defined, then $q$ is even. We extend the complement notation in the natural way to strings, namely, if $x=x_0 x_1 \dots x_{n-1}\in\Z_q^n$, then $\ox=\overline{x_0} \,\overline{x_1} \dots \overline{x_{n-1}}$.

Another important notation we introduce is the \emph{reverse} of a string. If $x=x_0 x_1 \dots x_{n-1}\in\Z_q^n$, then $x^R \eqdef x_{n-1} \dots x_1 x_0$. Obviously, the complement and reverse operations commute, and so $\overline{x^R} = \ox^R$.

The two major concepts we shall study in this paper are string-duplication channels, and codes that correct duplication errors introduced by the channel. We shall start by describing the former. We shall generally follow the notation set by~\cite{JaiFarSchBru17a}.

In a string-duplication channel, strings from $\Z_q^*$ are transmitted. However, the channel introduces errors by applying a (finite) number of duplication rules. Generally speaking, a \emph{duplication rule} is just a function $T:\Z_q^*\to\Z_q^*$. The set of all duplication rules the channel may apply is denoted by $\cT\subseteq \Z_q^*{}^{\Z_q^*}$. The choice of rules to contain in this set is motivated by the biological processes present in our in-vivo DNA storage system. We shall focus on two such sets which we define later (for more examples, see~\cite{JaiFarSchBru17a}).

Assume $u,v\in\Z_q^*$ are two strings. We say $v$ is a \emph{descendant} of $u$ if there exist $t\geq 0$ and $T_1,T_2,\dots,T_t\in\cT$ such that
\[
v= T_t(T_{t-1}( \dots T_1(u)\dots)).
\]
If the exact sequence of rules is immaterial, we just write
\[
u\der^t v,
\]
to signify this relation between $u$ and $v$. If $t=1$ we may omit it and write $u\der v$. If $t$ is unknown or irrelevant, we simply write $u\der^*v$.  For convenience, we always have $u\der^0 u$. The set of all descendants of $u$ is called the \emph{descendant cone of $u$}, and is denoted by
\[
D^*(u) \eqdef \set*{ v\in\Z_q^* ~:~ u\der^*v }.
\]
Thus, $D^*(u)$ is the reflexive transitive closure of $\cT$ acting on $u$.

Motivated by biological processes, in this paper we focus on two sets of duplication rules: the reverse-complement and the palindromic. For the first, we define reverse-complement duplication rules of the following form:
\[
T^{\rc}_{i,k}(x)\eqdef uv \ov^R w \qquad \text{if $x=uvw$, $\abs{u}=i$, $\abs{v}=k$.}
\]
Thus, $T^{\rc}_{i,k}$ takes a $k$-factor $v$ which follows an $i$-prefix $u$, and inserts a reverse-complement copy of $v$ following its original position. Note that the rule is not defined on strings of length shorter than $i+k$. We then define the set of \emph{reverse-complement duplication rules} as
\[
\cT^{\rc}_k \eqdef \set*{T^{\rc}_{i,k} ~:~ i\geq 0}.
\]
We denote derivation by $\der_k$, and the descendant cone by $D_k$, to emphasize the length of the duplication $k$.

\begin{example}
Assume $q=4$, $k=2$, and in $\Z_q$ we have $\overline{0}=3$ and $\overline{1}=2$. In the reverse-complement duplication setting,
\[
01103203 \der_2 011032\underline{10}03,
\]
where the underlined factor is the inserted reverse-complement duplication.
\end{example}

Along the same lines, we define the palindromic-duplication rules of the following form:
\[
T^{\pal}_{i,k}(x)\eqdef uv v^R w \qquad \text{if $x=uvw$, $\abs{u}=i$, $\abs{v}=k$,}
\]
and the set of all \emph{palindromic duplication rules} as
\[
\cT^{\pal}_k \eqdef \set*{T^{\pal}_{i,k} ~:~ i\geq 0}.
\]
The only difference between the palindromic and the reverse-complement duplication rules is that the latter also complement the duplicated symbols. We once again use $\der_k$ for derivation and $D_k$ for the descendant cone. To avoid cumbersome notation we infer from the context whether we use palindromic or reverse-complement duplication rules. In particular, Section~\ref{sec:rcbin1} and Section~\ref{sec:rclong} use reverse-complement duplication rules, whereas only Section~\ref{sec:pal} uses palindromic duplication rules.

\begin{example}
Assume $q=4$ and $k=2$. In the palindromic duplication setting,
\[
01103203 \der_2 011032\underline{23}03,
\]
where the underlined factor is the inserted palindromic duplication.
\end{example}

We now reach the second major concept -- error-correcting codes. 

\begin{definition}
An $(n,M;t)$ \emph{duplication-correcting code}, with respect to a set of duplication rules $\cT$, is a set $C\subseteq\Z_q^n$, of size $\abs{C}=M$, such that for any two distinct codewords, $c,c'\in C$,
\[
D^t(c)\cap D^t(c') = \emptyset.
\]
Here, $t$ denotes the number of duplication errors the code can correct, where $t\in\N\cup\set{*}$, and where $t=*$ denotes any number of errors.
\end{definition}

An important figure of merit associated with a code's efficiency, is its \emph{rate}. If $C$ is an $(n,M;t)$ duplication-correcting code over $\Z_q$, then its rate is defined as
\[
R(C) \eqdef \frac{1}{n}\log_q M.
\]
The largest size of such a code will be denoted by $A_q(n;t)$, and codes attaining this upper bound will be called \emph{optimal}. Finally, the \emph{coding capacity} is defined as
\[
R_q(t) \eqdef \limsup_{n\to\infty} \frac{1}{n}\log_q A_q(n;t).
\]

In the remainder of the paper, we shall be focusing on the reverse-complement duplication channel and the palindromic duplication channel, both with duplication length $k$, over $\Z_q$, and with an unbounded number of errors, i.e., $t=*$. We shall therefore try to determine or bound $A_q(n;*)^{\rc}_k$ and $R_q(*)^{\rc}_k$ for the reverse-complement duplication channel, and $A_q(n;*)^{\pal}_k$ and $R_q(*)^{\pal}_k$ for the palindromic duplication channel.

\section{Reverse-Complement Duplication of Length $k=1$}
\label{sec:rcbin1}

In this section we shall completely determine $A_q(n;*)_1^{\rc}$, the maximal size of a $q$-ary code of length $n$ that is capable of correcting any number of reverse-complement duplications of length $1$. Thus, we shall also find the coding capacity $R_q(*)_1^{\rc}$. Our method is constructive, and we will present codes of optimal size. The results depend on the alphabet size, and in particular, we shall divide our discussion to the case of binary alphabet, $q=2$, and the case of larger alphabets of even size, $q\geq 4$.

\subsection{The Binary Case, $q=2$}

We determine an exact condition for two binary strings to have a common descendant. As we shall show, this depends only on the first bit of the string. We start with a simple technical lemma.

\begin{lemma}
\label{lem:firstbitderivation}
Let $a\in\Z_2$ be a bit. Then
\[
a \der^*_1 a\oa w
\]
for any $w\in\Z_2^*$.
\end{lemma}

\begin{proof}
We can write
\[
a\oa w = a \oa^{m_1} a^{m_2} \oa^{m_3} a^{m_4} \dots b^{m_\ell},
\]
where $b$ is either $a$ or $\oa$, and where $m_1\geq 1$ for all $i$. We then start with $a$, and duplicate that $a$ $m_1$ times,
\[
a \der^{m_1}_1 a \oa^{m_1}.
\]
We then duplicate the last copy of $\oa$ a total of $m_2$ times, so
\[
a \der^{m_1}_1 a \oa^{m_1} \der^{m_2}_1 a \oa^{m_1} a^{m_2}.
\]
Repeating this process of duplicating the last bit to create a new run of $a$ or $\oa$ results in the desired derivation
\[
a \der^*_1 a\oa w.
\]
\end{proof}

We can now give our main characterization of strings with a common descendant.

\begin{theorem}
\label{th:bin1cone}
Let $x,y\in \Z_2^+$. Then
\[
D_1^*(x)\cap D_1^*(y)\neq \emptyset
\]
if and only if $x_0=y_0$.
\end{theorem}

\begin{proof}
For the first direction, assume $x_0\neq y_0$. Since no duplication can change the first bit of the string, it follows that all the strings in $D_1^*(x)$ start with $x_0$, whereas all those in $D_1^*(y)$ start with $y_0$. Hence, $D_1^*(x)\cap D_1^*(y)=\emptyset$.

For the second direction, assume $x_0=y_0=a$. Let us therefore write
\begin{align*}
x &= a x', & y &= a y'.
\end{align*}
We distinguish between two cases. For the first case, assume $y'=\varepsilon$, namely, $y=a$. In that case we have,
\[
x=a x' \der_1 a \oa x',
\]
and by Lemma~\ref{lem:firstbitderivation} we also have
\[
y = a \der_1^* a \oa x'.
\]
Hence, $D_1^*(x)\cap D_1^*(y)\neq\emptyset$.

For the second case, assume $y'\neq\varepsilon$, and write $y'=y''b$, with $b\in\Z_2$. We now observe the derivations
\begin{align*}
x&=ax' \der_1^* a \oa y'' b x' \der_1 a \oa y'' b \ob x', \\
y&=ay''b \der_1 a \oa y'' b \der_1^* a \oa y'' b \ob x',
\end{align*}
where both of the $\der_1^*$ derivations follow from Lemma~\ref{lem:firstbitderivation}.
\end{proof}

We thus immediately deduce the following:

\begin{corollary}
\label{cor:binary1}
For all $n\in\N$
\[
A_2(n;*)_1^{\rc} = 2,
\]
and therefore,
\[
R_2(*)_1^{\rc} = 0.
\]
\end{corollary}
\begin{proof}
Since in a binary $(n,M;*)_1^{\rc}$ code we need to have pair-wise disjoint descendant cones for distinct codewords, by Theorem~\ref{th:bin1cone}, we can have at most two codewords: one starting with a $0$, and one starting with a $1$. This is also easily achievable, e.g., the repetition code $C=\set{0^n,1^n}$. Hence, the maximal size of a binary code in this case is $A_2(n;*)_1^{\rc}=2$. The coding capacity follows by definition.
\end{proof}

\subsection{The Non-binary Case, $q\geq 4$}

In this section we assume a general non-binary alphabet, namely, $\Z_q$, with $q\geq 4$ even. We shall, once again, characterize the exact condition for two strings to have a joint descendant, thereby allowing us to find optimal codes as well as the coding capacity.

We introduce some new notation. Let $a\in\Z_q$ be a letter. We define
\[
a^\oplus \eqdef a\set*{a,\oa}^*,
\]
namely, $a^\oplus$ denotes the set of all strings starting with the letter $a$, and followed by any finite length string composed of $a$ and $\oa$ only.

If $x\in\Z_q^+$ is a non-empty string, we can always decompose it into
\begin{equation}
\label{eq:decomp}
x \in a_0^\oplus a_1^\oplus \dots a_{\ell-1}^\oplus,
\end{equation}
where $a_i\in\Z_q$ for all $i$. We say this is a minimal decomposition if $\ell$ is the smallest possible, thus necessarily, $a_{i+1}\not\in\set{a_i,\overline{a_i}}$ for all $i$. If~\eqref{eq:decomp} is a minimal decomposition of $x$, we define the \emph{signature of $x$} to be
\[
\sigma(x) \eqdef a_0 a_1 \dots a_{\ell-1}.
\]

\begin{example}
Assume $q=4$, $k=2$, and in $\Z_q$ we have $\overline{0}=3$ and $\overline{1}=2$. Then
\[
\sigma(0110300203) = 01020.
\]
This is because the minimal decomposition of $0110300203$ is given by $0|11|0300|2|03$, where the vertical lines show the relevant factors in the decomposition. Notice how
\[
0\in 0^\oplus, \quad 11\in 1^\oplus, \qquad 0300\in 0^\oplus, \qquad 2\in 2^\oplus, \qquad 03\in 0^\oplus.
\]
\end{example}

We first observe that all strings in a descendant cone share the same signature.

\begin{lemma}
\label{lem:signature}
Let $x\in\Z_q^+$ be a string. Then for all $x'\in D^*_1(x)$,
\[
\sigma(x)=\sigma(x').
\]
\end{lemma}

\begin{proof}
Let $x\in\Z_q^+$ be any string, and let its signature be $\sigma(x)=a_0 a_1 \dots a_{\ell-1}$, with $a_i\in\Z_q$ for all $i$. Furthermore, assume its minimal decomposition is
\[
x = x^{(0)} x^{(1)} \dots x^{(\ell-1)},
\]
where $x^{(i)}\in a_i^\oplus$ for all $i$.

Now let $x''$ be obtained from $x$ using a single reverse-complement duplication of length $1$, namely, $x\der_1 x''$. If the letter being duplicated resides in $x^{(i)}$, then the effect of the duplication is merely extending $x^{(i)}$ by one letter, without changing its first one. Thus,
\[
\sigma(x'')=\sigma(x).
\]
By a simple induction this may be extended to any number of derivations, and so for any $x'\in D_1^*(x)$ we have $\sigma(x)=\sigma(x')$.
\end{proof}

The signature can also determine whether two strings have a common descendant, as the following theorem shows.

\begin{theorem}
\label{th:nbin1cone}
Let $x,y\in\Z_q^+$. Then
\[
D_1^*(x)\cap D_1^*(y)\neq \emptyset
\]
if and only if $\sigma(x)=\sigma(y)$.
\end{theorem}

\begin{proof}
For the first direction, assume $\sigma(x)=\sigma(y)=a_0 a_1 \dots a_{\ell-1}$. That means we can write
\begin{align*}
x &= x^{(0)} x^{(1)} \dots x^{(\ell-1)},\\
y &= y^{(0)} y^{(1)} \dots y^{(\ell-1)},
\end{align*}
where $x^{(i)},y^{(i)}\in a_i^\oplus$, for all $i$. We also note that $x^{(i)}$ and $y^{(i)}$ begin with the same letter, $a_i$, and contain only letters from $\set{a_i,\overline{a_i}}$. Hence, by Theorem~\ref{th:bin1cone}, for each $i$, there exists $z^{(i)}\in\Z_q^+$ such that
\[
z^{(i)}\in D_1^* (x^{(i)})\cap D_1^* (y^{(i)}).
\]
Namely,
\[
x^{(i)} \der_1^* z^{(i)}, \qquad\text{and}\qquad
y^{(i)} \der_1^* z^{(i)}.
\]
It then follows that
\begin{align*}
x&= x^{(0)} x^{(1)} \dots x^{(\ell-1)} \der_1^* z^{(0)} z^{(1)} \dots z^{(\ell-1)},\\
y&= y^{(0)} y^{(1)} \dots y^{(\ell-1)} \der_1^* z^{(0)} z^{(1)} \dots z^{(\ell-1)},\\
\end{align*}
and so
\[
D_1^*(x)\cap D_1^*(y)\neq \emptyset.
\]

For the other direction, assume $D_1^*(x)\cap D_1^*(y)\neq \emptyset$, and let $z\in D_1^*(x)\cap D_1^*(y)$ be a common descendant of $x$ and $y$. By Lemma~\ref{lem:signature},
\[
\sigma(x) = \sigma(z) = \sigma(y),
\]
which completes the proof.
\end{proof}

We note that Theorem~\ref{th:nbin1cone} is a generalization of Theorem~\ref{th:bin1cone}, since in the binary case, the signature of any non-empty string is just its first letter.

By knowing when descendant cones intersect, we can now construct $q$-ary $(n,M;*)^{\rc}_1$ codes.

\begin{construction}
\label{con:qary}
Let $n\in\N$. We construct the following code:
\[
C = \set*{a_0 a_1\dots a_{\ell-1}a_{\ell-1}^{n-\ell} ~:~ 1\leq \ell\leq n, \forall i, a_i\in\Z_q, a_{i+1}\notin\set*{a_i,\overline{a_i}}}.
\]
\end{construction}

\begin{theorem}
\label{th:qaryopt}
The code $C$ from Construction~\ref{con:qary} is an optimal $(n, M; *)^{\rc}_{1}$ code, with
\[
M=q\cdot \frac{(q-2)^n-1}{q-3}.
\]
\end{theorem}
\begin{proof}
All the codewords of $C$ are of length $n$. Consider a codeword $c\in C$, with
\[
c=a_0 a_1\dots a_{\ell-1}a_{\ell-1}^{n-\ell},
\]
where for all $i$ we have $a_i\in\Z_q$ and $a_{i+1}\not\in\set{a_i,\overline{a_i}}$. By definition, the signature of $c$ is 
\[
\sigma(c) = a_0 a_1 \dots a_{\ell-1}.
\]
It follows that distinct codewords of $C$ have distinct signatures, and so by Theorem~\ref{th:nbin1cone}, they have disjoint descendant cones. Additionally, each possible signature appears in $C$, and so it is optimal. It remains to compute its size, $M$. Assuming we want a signature of length $\ell$, we have $q$ options for the first letters, and each subsequent letter has $q-2$ options. Thus,
\[
M=\sum_{\ell=1}^{n} q(q-2)^{\ell-1}=q\cdot \frac{(q-2)^n-1}{q-3}.
\]
\end{proof}

We comment that decoding the code from Construction~\ref{con:qary} is straightforward. Assume $z\in\Z_q$ was received. We then compute its signature, which we denote by $\sigma(z)=a_0 a_1 \dots a_{\ell-1}$, with $a_i\in\Z_q$ for all $i$. Then the decoding function returns $a_0 a_1 \dots a_{\ell-1} a_{\ell-1}^{n-\ell}\in C$. The entire process runs in $O(\abs{z})$ time, i.e., linear in the length of the received string.

Finally, we obtain the following:

\begin{corollary}
\label{cor:nonbinary1}
For all $n\in\N$ and even $q\in\N$, $q\geq 4$,
\[
A_q(n;*)^{\rc}_1 = q\cdot \frac{(q-2)^n-1}{q-3},
\]
and therefore
\[
R_q(*)^{\rc}_1 = \log_q (q-2).
\]
\end{corollary}

\begin{proof}
The claims follow directly from the size of the optimal code proved in Theorem~\ref{th:qaryopt}.
\end{proof}

\section{Reverse-Complement Duplication of Length $k\geq 2$}
\label{sec:rclong}

We shall now study codes that correct any number of reverse-complement duplication errors of length $k\geq 2$. We will show that the number of codewords in such a code is upper bounded by a constant, therefore implying a vanishing coding capacity.

Our proof strategy consists of identifying a property of strings that ensures the existence of a common descendant. We call this the \emph{$k$-summary} of the string. Since only a constant number of summaries are possible, the main claim will follow.

Let $w\in\Z_q^*$ be a string. We say that $a\in\Z_q$ is the \emph{$i$-th letter from the end} in $w$ if we can write
\[
w = w' a w''
\]
with $w',w''\in\Z_q^*$ and $\abs{w''}=i$. In particular, the $0$-th letter from the end of $w$ is the last letter of $w$. We also require the following technical tool from~\cite{BenSch22}:

\begin{lemma}[{{\cite[Lemma 2]{BenSch22}}}]
\label{lem:push}
Let $k\geq2$, $x \in \Z_q^*, |x|\geq k + 1$, and let $i \geq 2$ be an integer. If the $i$-th letter from the end of $x$ is $a$, then there exists $x'\in D^2_k(x)$, such that $a$ is the $(i-2)$-nd letter from the end of $x'$.
\end{lemma}

An important concept of subsequences is now defined.

\begin{definition}
\label{def:prop}
Let $k\geq 2$ be an integer, $x\in\Z_q^k$ be a string of length $k$, and let $y\in\Z_q^*$ be any string. We say $x$ is \emph{properly spaced} in $y$ if both of the following hold:
\begin{enumerate}
\item
$x$ is a subsequence of $y$, namely, there exist $y^{(0)}, y^{(1)}, y^{(2)},\dots,y^{(k)}\in\Z_q^*$ such that
\[
y = y^{(0)} x_0 y^{(1)} x_1 y^{(2)} \dots y^{(k-1)} x_{k-1} y^{(k)}.
\]
\item
If $k$ is even, then
\[
\abs*{y^{(1)}} \equiv \abs*{y^{(2)}} \equiv \dots \equiv \abs*{y^{(k)}} \equiv 0 \pmod{2}.
\]
\end{enumerate}
\end{definition}

\begin{example}
Consider the string $y=0110300203$. The string $323$ is properly spaced in $y$, as shown by the underlined locations $0110\underline{3}00\underline{2}0\underline{3}$ (note that $k=\abs{323}=3$ is odd in this case). The string $32$ is also properly spaced in $y$, as we see in $0110\underline{3}00\underline{2}03$ (this time $k=\abs{32}=2$ is even). However, $23$ is not properly spaced in $y$, since the length of the factor between $2$ and $3$ in $y$ is odd, and $k=\abs{23}=2$ is even.
\end{example}

We observe that the property of being properly-spaced persists in all descendants.

\begin{lemma}
\label{lem:propdes}
Let $k\geq 2$, $x\in\Z_q^k$, and $y\in\Z_q^*$, such that $x$ is properly spaced in $y$. If $y'\in D^*_k(y)$, then $x$ is also properly spaced in $y'$.
\end{lemma}

\begin{proof}
Since $y'$ is obtained from $y$ by insertions only, the first requirement of Definition~\ref{def:prop} certainly holds for $x$ in $y'$. Additionally, if $k$ is even, since duplications insert $k$-factors, it follows that the parities of all $y^{(i)}$ from Definition~\ref{def:prop} remain unchanged.
\end{proof}

When a string is properly spaced, we can extract it as a suffix in one of the descendants, as the following lemma shows.

\begin{lemma}
\label{lem:suffix}
Let $k\geq 2$, $x\in\Z_q^k$, and $y\in\Z_q^*$, such that $x$ is properly spaced in $y$. Then there exists $v\in\Z_q^*$ such that
\[
y \der^*_k vx.
\]
\end{lemma}

\begin{proof}
By definition, we necessarily have $\abs{y}\geq k$. We consider two cases, depending on the parity of $k$.

If $k$ is even, the last letter of $x$, namely, $x_{k-1}$, is at an even distance from the end. If it is at the end, we duplicate the $k$-suffix twice (this does not change the last letter, but ensures at least one duplication was employed). Otherwise, we can use Lemma~\ref{lem:push} repeatedly to get $x_{k-1}$ to the end, with each application the distance of $x_{k-1}$ to the end diminishes by $2$. Thus, we have
\[
y \der^*_k v^{(k-1)} x_{k-1},
\]
for some $v^{(k-1)}\in \Z_q^*$. Since at least one duplication was used, we have
\[
\abs*{v^{(k-1)}x_{k-1}} \geq \abs{y}+k \geq 2k\geq k+2.
\]
By Lemma~\ref{lem:propdes}, $x$ is also properly spaced in $v^{(k-1)}x_{k-1}$. That means $x_{k-2}$ is at an even distance from the end of $v^{(k-1)}$. We ignore the last letter, $x_{k-1}$, and repeat the process above to push the letter $x_{k-2}$:
\[
v^{(k-1)}x_{k-1} \der^*_k v^{(k-2)}x_{k-2}x_{k-1}.
\]
Continuing in this manner for all the letters of $x$ we finally obtain
\[
y \der^*_k vx
\]
for some $v\in\Z_q^*$.

If $k$ is odd, we repeat the same proof as in the case of $k$ even. The only difference is that we are not guaranteed that the letter we are trying to push, $x_i$, is at an even distance from the end. If that is not the case, and the distance is odd, we simply perform a duplication on the suffix, thus making the distance even.
\end{proof}

Another technical helpful tool allows us to synchronize derivations in certain cases, and produce a common descendant.

\begin{lemma}
\label{lem:sync}
Let $k\geq 2$, $x\in\Z_q^k$, and $y\in\Z_q^*$, such that $x$ is properly spaced in $y$. Then there exists $v\in\Z_q^*$ such that
\begin{align*}
y & \der^*_k vx, & yx \der^*_k vx.
\end{align*}
\end{lemma}

\begin{proof}
By Lemma~\ref{lem:suffix}, there exists $u\in\Z_q^*$ such that
\begin{equation}
\label{eq:sync}
y \der^*_k ux.
\end{equation}
Continuing this derivation we get
\[
y \der^*_k ux \der^2_k ux\ox^R x = vx,
\]
where we define $v\eqdef ux\ox^R$.

Now, starting with the same derivation as in~\eqref{eq:sync}, we get
\[
yx \der^*_k uxx \der_k ux\ox^R x = vx,
\]
which completes the proof.
\end{proof}

The following definition provides the property used by our main theorem in this section.

\begin{definition}
Let $k\in\N$, $k\geq 2$, let $\ell\geq 0$ be an integer, and let $y\in\Z_q^{\ell k}$ be a string. Partition $y$ into $k$-factors,
\[
y = y^{(0)} y^{(1)} \dots y^{(\ell-1)},
\]
with $y^{(i)}\in\Z_q^k$ for all $i$. Define
\[
\phi(y^{(i)}) \eqdef \begin{cases}
\varepsilon & \text{if there exists $i'<i$ such that $y^{(i')}=y^{(i)}$,}\\
y^{(i)} & \text{otherwise.}
\end{cases}
\]
Then the \emph{summary} of $y$, denoted $\Phi(y)$ is defined as
\[
\Phi(y) \eqdef \phi(y^{(0)}) \phi(y^{(1)}) \dots \phi(y^{(\ell-1)}).
\]
\end{definition}

Intuitively, to compute the summary of string, we scan its $k$-factors from left to right. Each $k$-factor that appears for the first time is appended to the summary, whereas $k$-factors that have already appeared before are removed. It follows that the summary of a string contains a partition into distinct $k$-factors, and whose total length is upper bounded by $kq^k$.

\begin{example}
Assume $q=4$ and $k=2$. Let $y=011011013030023003$. Its partition into $2$-factors is $01|10|11|01|30|30|02|30|03$. If we underline repeated $2$-factors in the partition we get $01|10|11|\underline{01}|30|\underline{30}|02|\underline{30}|03$. By removing the repeated $2$-factors from the partition, we obtain the summary of $y$,
\[
\Phi(y)=011011300203.
\]
We observe that while $01$ appears twice as a $2$-factor in the summary, $\underline{01}1\underline{01}1300203$, the second appearance is not removed since it does not appear in the partition.
\end{example}

In the main theorem for this section, we show that strings which have the same summary (up to a short prefix perhaps) also have a common descendant.

\begin{theorem}
\label{th:longk}
Let $k\geq 2$, and $x,y\in\Z_q^*$. Write
\begin{align*}
x &= x' x'' & y&= y' y'',
\end{align*}
where
\begin{align*}
\abs*{x'} & = \abs*{x} \bmod k, & \abs*{y'}=\abs*{y} \bmod k.
\end{align*}
If $x'=y'$ and $\Phi(x'')=\Phi(y'')$, then
\[
D^*_k(x)\cap D^*_k(y)\neq \emptyset.
\]
\end{theorem}

\begin{proof}
We shall actually prove that $x''$ and $y''$ have a common descendant, namely,
\[
D^*_k(x'')\cap D^*_k(y'') \neq \emptyset.
\]
Then, by prepending the prefix $x'=y'$, we shall trivially deduce the desired claim,
\[
D^*_k(x'x'')\cap D^*_k(y'y'') \neq \emptyset.
\]
Thus, we shall forget about the prefix $x'=y'$ from now on. Other simple cases we easily dismiss are that of $x=y$, as well as the case of $\Phi(x'')=\Phi(y'')=\varepsilon$, since then, again, $x=y$.

Let us now partition $x''$ and $y''$ into $k$-factors,
\begin{align*}
x'' & = x^{(0)} x^{(1)} \dots x^{(\ell_x-1)}, &
y'' & = y^{(0)} y^{(1)} \dots y^{(\ell_y-1)},
\end{align*}
where $x^{(i)}, y^{(j)}\in\Z_q^k$ for all $i$ and $j$. The remainder of the proof shall proceed in iterations. We initialize two indices, $i_x=i_y=0$. We shall, at each iteration, make sure that
\begin{equation}
\label{eq:req1}
\Phi(x^{(0)} \dots x^{(i_x-1)}) = \Phi(y^{(0)} \dots y^{(i_y-1)}),
\end{equation}
and that there exists some $u\in\Z_q^*$ such that
\begin{align}
\label{eq:req2}
x^{(0)} \dots x^{(i_x-1)} &\der^*_k u, &
y^{(0)} \dots y^{(i_y-1)} &\der^*_k u.
\end{align}
Obviously, before we begin, when $i_x=i_y=0$, both~\eqref{eq:req1} and~\eqref{eq:req2} hold, since for~\eqref{eq:req1} we have $\Phi(\varepsilon)=\varepsilon$, and for~\eqref{eq:req2} choosing $u=\varepsilon$ gives the trivial $\varepsilon\der^*_k \varepsilon$. We distinguish between the following cases:

\textbf{Case 1:} $i_x<\ell_x$, $i_y<\ell_y$ and $x^{(i_x)}=y^{(i_y)}$. Since~\eqref{eq:req1} holds, we also have,
\[
\Phi(x^{(0)} \dots x^{(i_x)}) = \Phi(y^{(0)} \dots y^{(i_y)}).
\]
Additionally, since~\eqref{eq:req2} holds, we obviously have
\begin{align*}
x^{(0)} \dots x^{(i_x-1)} x^{(i_x)} &\der^*_k u x^{(i_x)}, &
y^{(0)} \dots y^{(i_y-1)} y^{(i_y)} &\der^*_k u y^{(i_y)}.
\end{align*}
Thus, we can increase both $i_x$ and $i_y$ by $1$, while maintaining~\eqref{eq:req1} and~\eqref{eq:req2}.

\textbf{Case 2:} $i_x<\ell_x$, $i_y<\ell_y$ and $x^{(i_x)}\neq y^{(i_y)}$. Since~\eqref{eq:req1} holds, and also $\Phi(x'')=\Phi(y'')$, there exists $i'_x<i_x$ such that $x^{(i'_x)}=x^{(i_x)}$, or there exists $i'_y<i_y$ such that $y^{(i'_y)}=y^{(i_y)}$. Let us assume the former (if the latter holds, the proof is symmetric). We then have
\[
\Phi(x^{(0)} \dots x^{(i_x)}) =
\Phi(x^{(0)} \dots x^{(i_x-1)}) =
\Phi(y^{(0)} \dots y^{(i_y-1)}).
\]

By definition, we observe that $x^{(i_x)}=x^{(i'_x)}$ is properly spaced in $x^{(0)}\dots x^{(i_x-1)}$. By Lemma~\ref{lem:propdes}, this is true also in any descendant of $x^{(0)}\dots x^{(i_x-1)}$, in particular, $u$. By Lemma~\ref{lem:sync}, there exists $u'\in\Z_q^*$ such that
\begin{align*}
u &\der^*_k u', & u x^{(i_x)} &\der^*_k u'.
\end{align*}
We therefore have
\begin{align*}
x^{(0)} \dots x^{(i_x-1)} x^{(i_x)} &\der^*_k u x^{(i_x)}\der^*_k u', &
y^{(0)} \dots y^{(i_y-1)} &\der^*_k u \der^*_k u'.
\end{align*}
It follows that we can increase $i_x$ by $1$, while maintaining~\eqref{eq:req1} and~\eqref{eq:req2}.

\textbf{Case 3:} $i_x < \ell_x$ and $i_y = \ell_y$. In this case we have $y^{(0)}\dots y^{(i_y-1)}=y''$. Since $\Phi(y'')=\Phi(x'')$, there exists $i'_x<i_x$ such that $x^{(i'_x)}=x^{(i_x)}$. The remainder of this case is the same as Case 2.

\textbf{Case 4:} $i_x = \ell_x$ and $i_y < \ell_y$. This case is symmetric to Case 3.

\textbf{Case 5:} $i_x=\ell_x$ and $i_y=\ell_y$. When we reach this case, $u$ from~\eqref{eq:req2} is the desired common descendant $x''$ and $y''$.

Finally, we observe that we shall always reach Case 5, since in each of the other cases, $i_x$ or $i_y$ is increased by $1$. This completes the proof.
\end{proof}

With the help of Theorem~\ref{th:longk} we obtain the following corollary.

\begin{corollary}
\label{cor:long}
For any $k\in\N$, $k\geq 2$, any $n\in\N$ and any even $q\in\N$,
\[
A_q(n;*)^{\rc}_k \leq q^{k-1}\cdot \frac{q^{kq^k+k}-1}{q^k-1},
\]
and therefore
\[
R_q(*)^{\rc}_k = 0.
\]
\end{corollary}

\begin{proof}
By Theorem~\ref{th:longk}, a code cannot contain two codewords with the same summary and the same prefix of length at most $k-1$. The number of prefixes is upper bounded by $q^{k-1}$. As for the summaries, each summary is a string of length $\ell k$, for some $0\leq \ell\leq q^k$. Thus,
\[
A_q(n;*)^{\rc}_k \leq q^{k-1} \sum_{\ell=0}^{q^k} q^{k\ell} = q^{k-1}\cdot \frac{q^{kq^k+k}-1}{q^k-1}.
\]
Substituting this into the definition of coding capacity we get
\[
R_q(*)^{\rc}_k = 0,
\]
as claimed.
\end{proof}

\section{Palindromic Duplication}
\label{sec:pal}

Having determined the coding capacity of the reverse-complement duplication channel in the previous sections, we now turn to the simpler case of palindromic duplication. In this case, factors are duplicated in reverse, but without applying the complement.

We first look at the case of duplication of length $k=1$. In this case, reversing does nothing, and so this duplication is the same as tandem duplication of length $k=1$, which was already studied in~\cite{JaiFarSchBru17a}.

\begin{corollary}
For any $n,q\in\N$, $q\geq 2$,
\[
A_q(n;*)^{\pal}_1 = q\sum_{i=0}^{n-1}(q-1)^i =
\begin{cases}
q \cdot \frac{(q-1)^n-1}{q-2} & q\geq 3, \\
qn & q=2,
\end{cases}
\]
and therefore
\[
R_q(*)^{\pal}_1 = \log_q (q-1).
\]
\end{corollary}
\begin{proof}
Since a palindromic duplication of length $1$ is the same as a tandem duplication of length $1$, we look at the known results for tandem duplication. According to~\cite[Theorems 15,16]{JaiFarSchBru17a}, the optimal codes in these cases contain all the strings whose adjacent positions contain distinct letters, except perhaps in the suffix. Choosing the first letter has $q$ options, followed by $0\leq i\leq n-1$ letters distinct from their predecessors in the string for a total of $(q-1)^i$ options, finally followed by repeating the last letter chosen to get a string of length $n$.
\end{proof}

We note that the results are quite similar to Corollaries~\ref{cor:binary1} and~\ref{cor:nonbinary1}. In the binary case the coding capacity vanishes in both types of duplication, whereas in the non-binary case, it is slightly higher in palindromic duplication, compared with reverse-complement duplication.

Moving on to longer duplication lengths, $k\geq 2$, we observe that the complement operation plays no role whatsoever in Section~\ref{sec:rclong}. We do need to replace Lemma~\ref{lem:push}, which we do easily by adapting \cite[Lemma 2]{BenSch22}.

\begin{lemma}
\label{lem:palpush}
Let $k\geq2$, $x \in \Z_q^*, |x|\geq k + 1$, and let $i \geq 2$ be an integer. If the $i$-th letter from the end of $x$ is $a$, then there exists $x'\in D^{2}_k(x)$, such that $a$ is the $(i-2)$-nd letter from the end of $x'$, and where $D(\cdot)$ denotes descendants with respect to palindromic duplication.
\end{lemma}

\begin{proof}
Let us write $x=uvbcw$, where $u,v,w\in\Z_q^*$, $\abs{v}=k-1$, $b,c\in\Z_q$, and where $a$, the $i$-th letter from the end of $x$, occurs in $v$. Then,
\[
x=uvbcw \der_k uvbbv^R cw \der_k uvbbv^R ccvw \eqdef x',
\]
where $\der$ denotes derivation with respect to palindromic duplication. Thus, $x'$ contains $a$ as the $(i-2)$-nd letter from the end.
\end{proof}

We now obtain the palindromic analogue of Corollary~\ref{cor:long}.

\begin{corollary}
\label{cor:pallong}
For any $k\in\N$, $k\geq 2$, any $n\in\N$ and any even $q\in\N$,
\[
A_q(n;*)^{\pal}_k \leq q^{k-1}\cdot \frac{q^{kq^k+k}-1}{q^k-1},
\]
and therefore
\[
R_q(*)^{\pal}_k = 0.
\]
\end{corollary}

\begin{proof}
Scanning Section~\ref{sec:rclong}, we replace Lemma~\ref{lem:push} with Lemma~\ref{lem:palpush}, and in Lemma~\ref{lem:sync} we simply remove the complement. The rest remains the same, giving us the same conclusion in the palindromic case as in the reverse-complement case.
\end{proof}

\section{Conclusion}
\label{sec:conc}

In this work, we studied the parameters of duplication-correcting codes, both in the reverse-complement and the palindromic settings. We determined $A_q(n;*)^{\rc}_1$ and constructed optimal codes. We then showed that the coding capacity, $R_q(n;*)^{\rc}_k$, $k\geq 2$, is vanishing. A similar result for $k\geq 2$ and palindromic duplication was also proved.

Other open questions remain. While having a vanishing coding capacity is disappointing, we do not yet know how to construct $(n,M;t)^{\rc}_k$ and $(n,M;t)^{\pal}_k$ for a finite $t$. If $t=1$, namely, a single duplication error is to be corrected, then~\cite{BenSch22} constructed reverse-complement duplication-correcting codes when $k$ is odd. More generally, we can use a burst-insertion-correcting code, e.g., \cite{SchWacGabYaa17,BitHanPolVor21,WanTanSimGabFar24}. However, when $t\geq 2$ no solution is known except using a $tk$-insertion correcting code, which is, most likely, far from optimal. We leave finding such codes, determining the optimal code size as well as the coding capacity, for future work.

\backmatter

\bmhead{Acknowledgments}

This work was supported in part by the US National Science Foundation (NSF), grant CCF-5237372, and by the Zhejiang Lab BioBit Program (grant no. 2022YFB507). The author M.~Schwartz is currently on a leave of absence from Ben-Gurion University of the Negev.
    
\bibliography{allbib}

\end{document}